\documentclass{article}

\usepackage{amsmath}

\usepackage{mathrsfs, amssymb,amsmath,url,amsthm,comment}

\theoremstyle{plain}

    \newtheorem{thm}{Theorem}[section]
    \newtheorem{theorem}[thm]{Theorem}
    \newtheorem{lemma}[thm]{Lemma}
    \newtheorem{proposition}[thm]{Proposition}
    
    \newtheorem{question}[thm]{Question}

\theoremstyle{definition}

    \newtheorem{definition}[thm]{Definition}

\DeclareMathOperator{\CSP}{CSP}
\DeclareMathOperator{\PCSP}{PCSP}
\DeclareMathOperator{\OneInThree}{1-in-3}
\DeclareMathOperator{\NAE}{NAE-3}
\DeclareMathOperator{\OneInThreeSAT}{1-in-3-SAT}
\DeclareMathOperator{\NAESAT}{NAE-3-SAT}

\DeclareMathOperator{\arity}{arity}

\newcommand{\relstr}[1]{\mathbb{#1}}
\newcommand{\ttt}[1]{\langle #1 \rangle}
\newcommand{\rrr}[1]{[#1]}
\newcommand{\vc}[1]{\mathbf{#1}}
\newcommand{\abs}[1]{|#1|}

\begin{document}

\title{Promises Make Finite (Constraint Satisfaction) Problems Infinitary}

\author{Libor Barto
\thanks{
Libor Barto has received funding from the European Research Council
(ERC) under the European Unions Horizon 2020 research and
innovation programme (grant agreement No 771005).
} \\
Department of Algebra\\Faculty of Mathematics and Physics\\
Charles University\\
Prague, Czechia\\
Email: libor.barto@gmail.com}

\date{April 1, 2019}

\maketitle

\begin{abstract}
The fixed template Promise Constraint Satisfaction Problem (PCSP) is a recently proposed significant generalization of the fixed template CSP, which includes approximation variants of satisfiability and graph coloring problems. All the currently known tractable (i.e., solvable in polynomial time) PCSPs over finite templates can be reduced, in a certain natural way, to tractable CSPs. However, such CSPs are often over infinite domains. We show that the infinity is in fact necessary by proving that a specific finite-domain PCSP, namely (1-in-3-SAT, Not-All-Equal-3-SAT), cannot be naturally reduced to a tractable finite-domain CSP, unless P=NP. 
\end{abstract}


\section{Introduction}


Finding a 3-coloring of a graph or finding a satisfying assignment of a propositional 3-CNF formula  (or rather the decision variants of these problems) are well-known and fundamental NP-complete computational problems. The latter problem, 3-SAT, has many restrictions still known to be NP-complete~\cite{Sch78}, two of which play a central role in this paper. The \emph{positive 1-in-3-SAT}, denoted $\OneInThreeSAT$, can be defined as follows. The instance is a list of triples of variables and the problem is to find a mapping from the set of variables to $\{0,1\}$ such that exactly one variable in each triple is assigned 1. In the \emph{positive Not-All-Equal-3-SAT}, denoted $\NAESAT$, instances are triples of variables as well, but the mapping is only required to assign  not-all-equal elements to each triple.  

There are two ways how to relax the requirement on the assignment in order to get a potentially simpler problem.
The first one is to require a specified fraction of the constraints to be satisfied. For example, given a satisfiable 3-SAT instance, is it easier to find an assignment satisfying at least 90\% of clauses? A celebrated result of H{\aa}stad~\cite{H01} proves that the answer is ``No.'' -- it is still an NP-complete problem. (Actually, any fraction greater than $7/8$ gives rise to an NP-complete problem while the fraction $7/8$ is achievable in polynomial time.) 

The second type of relaxation is to require that a specified weaker version of every constraint is satisfied. For example, we want to find a 100-coloring of a 3-colorable graph, or we want to find a valid $\NAESAT$ assignment to a $\OneInThree$-satisfiable instance.
The complexity of the former problem is a notorious open question (for a recent development see~\cite{BG16,BKO18}, but even 6-coloring a 3-colorable graph is not known to be NP-complete).
On the other hand, the latter problem admits an elegant polynomial time algorithm~\cite{BG18a,BG18b}, which we now describe.

We take a satisfiable instance of $\OneInThreeSAT$ and 
replace each triple of variables $(x,y,z)$ in the instance by the linear equation $x+y+z=1$ over $\mathbb{Z}$ (the integers). 
The obtained system is solvable (by the original 0,1 assignment) and it is known that finding a solution to a system of linear equations over $\mathbb{Z}$ is in P (see~\cite{GLS93}). Now, if $\phi$ is any solution to the system, then 
$$
\psi(x)= \left\{
  \begin{array}{ll}
	0 & \mbox{ if } \phi(x) \leq 0 \\
	1 & \mbox{ if } \phi(x) > 0
	\end{array}
	\right.
$$
is a valid $\NAESAT$ assignment.

Alternatively, one can solve the system over $\mathbb{Q} \setminus \{1/3\}$ by a simple adjustment of Gaussian elimination and define  $\psi(x)=0$ iff $\phi(x) < 1/3$. (A more general class of problems can be solved, e.g., by restricting the domain $\mathbb{Q} \setminus \{c\}$ to the interval $[0,1]$ and using an adjustment of linear programming rather than Gaussian elimination; see~\cite{BG18a,BG18b}.)

It is remarkable that both polynomial algorithms transfer the original problem over a finite domain to a problem over an infinite domain. The main result of this paper shows that this finite-to-infinite transition is unavoidable. 
This result is stated more precisely below, as Theorem~\ref{thm:main_fake}, but let us first describe its background. 

\subsection{Constraint Satisfaction Problems}

It will be convenient in this paper to use a formalization of CSP and PCSP  via homomorphisms of relational structures. We refer to~\cite{BKW17,BG18b} for translations to the other standard definitions. 

Let $\relstr{A}$ be a relational structure of a finite signature, often called \emph{template} in this context.
The \emph{Constraint Satisfaction Problem (CSP) over} $\relstr{A}$, denoted $\CSP(\relstr{A})$, is the problem of deciding whether a given finite relational structure $\relstr{X}$ (similar to $\relstr{A}$) has a homomorphism to $\relstr{A}$. The \emph{search problem} for $\CSP(\relstr{A})$ is to find such a homomorphism. Examples of CSPs include the 3-coloring problem (where $\relstr{A}$ is a structure with a three-element domain and the binary disequality relation), 3-SAT (where $\relstr{A}$ consists of 8 ternary relations of the form $(\neg) x \vee (\neg) y \vee (\neg) z$ on the domain $\{0,1\}$), the problems $\OneInThreeSAT$, $\NAESAT$ for which the templates are structures with domain $\{0,1\}$ and a single ternary relation
\begin{align*}
\OneInThree &= (\{0,1\};  \{(1,0,0),(0,1,0),(0,0,1)\}) \\
\NAE &= (\{0,1\}; \{0,1\}^3 \setminus \{(0,0,0),(1,1,1)\})\enspace,
\end{align*}
and the infinite-domain CSPs over $(\mathbb{Z}; x+y+z=1)$ and $(\mathbb{Q}\setminus\{1/3\};x+y+z=1)$ that were used to solve the relaxed version of $\OneInThreeSAT$. 

The complexity of the CSP over finite templates (modulo polynomial time reductions) is fully classified by a recent deep result of Bulatov~\cite{Bul17} and, independently, Zhuk~\cite{Zhuk17}. The classification is a culmination of an active research program, so called algebraic approach to CSPs, inspired by the landmark paper of Feder and Vardi~\cite{FV98}, where the authors conjectured that each finite template CSP is either tractable or NP-complete, and observed that the tractability is often tied to closure properties 
of relations in the template. General theory of CSPs, whose basics were developed in~\cite{JCG97,J98,BJK05,BOP18}, confirmed this observation by closely linking CSPs to Universal Algebra (the theory of general algebraic systems) and provided guidance and tools for the eventual resolution of the dichotomy conjecture in~\cite{Bul17,Zhuk17}. 

The theory of CSPs is based on a connection between constructions on relational structures (that lead to polynomial time reductions) and properties of their \emph{polymorphisms} -- multivariate functions preserving the structures. So far, the most general relational construction introduced in~\cite{BOP18} is the so called pp-construction. Roughly, we say that $\relstr{A}$ \emph{pp-constructs} $\relstr{B}$ if $\relstr{B}$ can be obtained from $\relstr{A}$ by a sequence of first-order interpretations restricted to primitive positive formulae and replacements by homomorphically equivalent structures (see Section~\ref{sec:prelim} for a more detailed definition). 
In this situation, there is a natural gadget reduction of $\CSP(\relstr{B})$ to $\CSP(\relstr{A})$.
It follows, for instance, that $\CSP(\relstr{A})$ is NP-complete whenever $\relstr{A}$ pp-constructs a template of 3-SAT.
It turned out~\cite{Bul17,Zhuk17}, confirming the \emph{algebraic dichotomy conjecture} from~\cite{BJK05}, that this is exactly the borderline between tractable and NP-complete CSPs: all templates that do not pp-construct the template of 3-SAT have tractable CSPs. 
The algebraic part of the theory will not be discussed here, let us just mention that the strongest available algebraic 
characterization of the borderline by means of cyclic operations~\cite{BK12} is essential for the proof of the main result, Theorem~\ref{thm:main_fake}.

\subsection{Promise CSPs}

A \emph{template} for the \emph{Promise CSP} (PCSP) is a pair $(\relstr{A},\relstr{B})$ of similar relational structures of finite  signature such that $\relstr{A}$ has a homomorphism to $\relstr{B}$. The PCSP over such a pair, denoted $\PCSP(\relstr{A},\relstr{B})$, is the following promise problem: given a relational structure $\relstr{X}$ (similar to $\relstr{A}$ and $\relstr{B}$) output ``Yes.'' if $\relstr{X}$ has a homomorphisms to $\relstr{A}$, and ``No.'' if $\relstr{X}$ does not have a homomorphism to $\relstr{B}$. 
The \emph{search problem} for $\PCSP(\relstr{A},\relstr{B})$ is to find a homomorphism $\relstr{X} \to \relstr{B}$ given an input structure $\relstr{X}$ that has a homomorphism to $\relstr{A}$. Notice that $\CSP(\relstr{A})$ is the same as $\PCSP(\relstr{A},\relstr{A})$. Examples of PCSPs include the 100-coloring of a 3-colorable graph, where the two structures are $(\{1,2,3\}, \neq)$ and $(\{1, 2, \dots, 100\}, \neq)$, and $\PCSP(\OneInThree,\NAE)$ -- the central computational problem in this paper.
Similar examples were the motivation for introducing the PCSP framework in~\cite{AGH17,BG16,BG18a,BG18b}.

The current knowledge of the complexity of finite-domain PCSPs beyond CSPs is very much limited.
For example, the Feder-Vardi dichotomy conjecture for CSPs was inspired by two earlier classification results: for CSPs over a Boolean (i.e., two-element) domain~\cite{Sch78} and for CSPs over graphs~\cite{HN90}. In PCSPs, even the analogues of these early results are challenging. PCSPs over graphs include, as a very special case, the PCSP over a pair of complete graphs, the problem of $l$-coloring a $k$-colorable graph. 
A systematic study of Boolean PCSPs (where both structures have a two-element domain) was initiated in~\cite{BG18a}, 
but the general Boolean case is still wide open.

Fortunately, building on the initial insights and results in~\cite{AGH17,BG16,BG18a,BG18b}, it was observed in~\cite{BKO18} (among many other important results such as the NP-hardness of 5-coloring a 3-colorable graph) that the basics of the CSP theory from~\cite{BOP18} generalize to PCSPs. In particular, the notions of pp-constructions and polymorphisms have their PCSP counterparts and the connection between relational and algebraic structures works just as well as in the CSP. 
This is especially interesting because some hardness and algorithmic results in PCSP require techniques used in approximation. PCSP thus might help building a bridge between the discrete, universal algebraic world of (exact) CSPs and analytical world of approximation. 

\subsection{Finite PCSPs are infinitary}

The main result of this paper says that it is impossible to reduce $\PCSP(\OneInThree, \NAE)$ to a tractable finite-domain CSP by means of a pp-construction, unless P=NP.

\begin{theorem} \label{thm:main_fake}
Let $\relstr{C}$ be a finite relational structure that pp-constructs $(\OneInThree, \NAE)$.
Then $\CSP(\relstr{C})$ is NP-complete.
\end{theorem}

A fundamental question is whether each finite tractable PCSP template can be pp-constructed from an infinite tractable CSP template.
In~\cite{BG18b}, the authors conjectured that the answer is positive and even suggested a family of tractable CSPs that might solve all Boolean PCSPs. 
 
The class of all infinite-domain CSPs is very broad. In fact, each computational problem is equivalent to an infinite-domain CSP~\cite{BodG08}. However, some parts of the CSP theory can be extended to a quite rich class of structures, namely, reducts of finitely bounded homogeneous structures in finite signature, and some general results even to the broader class of $\omega$-categorical structures~\cite{Bod08,Pin15}. In particular, an algebraic criterion for NP-hardness is available~\cite{BarP16}, so it might be possible to generalize Theorem~\ref{thm:main_fake} to this setting (possibly with a different template). As $\omega$-categorical structures are, in a sense, close to finite and CSPs over them are solved by ``finitary'' algorithms, such a generalization would show that a polynomial time algorithm for some PCSP must be ``truly'' infinitary. 

Let us make a final remark before starting with the technicalities. 
Both algorithms~\cite{Bul17,Zhuk17} for the finite-domain CSP are extremely complex and simplifications are much desired. Theorem~\ref{thm:main_fake} supports the intuition that a simpler algorithm may require infinitary methods, such as CSPs over numerical domains~\cite{BodM17} ($\mathbb{Z}$, $\mathbb{Q}$, \dots). 


\section{Preliminaries} \label{sec:prelim}

In this section we give formal definitions of the concepts essential for the proof.  For an in depth introduction to CSP and PCSP, see \cite{BKW17,BKO18} and references therein. 

\subsection{PCSP}

A \emph{relational structure (of finite signature)} is a tuple $\relstr{A} = (A; R_1, \dots, R_n)$ where each $R_i \subseteq A^{\arity(R_i)}$ is a relation on $A$ of arity $\arity(R_i) \geq 1$. The structure $\relstr{A}$ is \emph{finite} if $A$ is finite.

Two relational structures $\relstr{A} = (A; R_1, \dots, R_n)$ and $\relstr{B} = (B; S_1, \dots, S_n)$ are \emph{similar} if
they have the same number of relations and $\arity(R_i) = \arity(S_i)$ for each $i \in \{1, \dots, n\}$. 

For two such similar relational structures $\relstr{A}$ and $\relstr{B}$, a \emph{homomorphism} from $\relstr{A}$ to $\relstr{B}$ is a mapping $f: A \to B$ such that $(f(a_1), f(a_2), \dots, f(a_k)) \in S_i$ whenever $i \in \{1, \dots, n\}$ and $(a_1, a_2, \dots, a_k) \in R_i$ where $k = \arity(R_i)$.

We write $\relstr{A} \to \relstr{B}$ if there exists a homomorphism from $\relstr{A}$ to $\relstr{B}$, and
$\relstr{A} \not\to \relstr{B}$ if there is none. 

\begin{definition}
A \emph{PCSP template} is a pair $(\relstr{A},\relstr{B})$ of similar relational structures such that $\relstr{A} \to \relstr{B}$.

The decision version of PCSP over $(\relstr{A},\relstr{B})$, written $\PCSP(\relstr{A},\relstr{B})$, is the following promise problem. Given a finite structure $\relstr{X}$ similar to $\relstr{A}$ (and $\relstr{B}$), output ``Yes.'' if $\relstr{X} \to \relstr{A}$ and output ``No.'' if $\relstr{X} \not\to \relstr{B}$. 

The search version of $\PCSP(\relstr{A},\relstr{B})$ is, given a structure $\relstr{X}$ similar to $\relstr{A}$ such that $\relstr{X} \to \relstr{A}$, find a homomorphisms $\relstr{X} \to \relstr{B}$. 
\end{definition}

In the case $\relstr{A}=\relstr{B}$ we talk about a CSP template (and simply write $\relstr{A}$ instead of $(\relstr{A},\relstr{A})$) and define $\CSP(\relstr{A}) = \PCSP(\relstr{A},\relstr{A})$. 

The decision version of the PCSP over $(\relstr{A},\relstr{B})$ can be reduced to the search version. For CSPs, it is known~\cite{BJK05} that these two versions are in fact equivalent, but it is an open problem whether they are equivalent for PCSPs as well.

\subsection{Constructions} \label{subsec:constr}

The two ingredients of a pp-construction are pp-powers and homomorphic relaxations. 

Homomorphic relaxation, called \emph{homomorphic sandwiching} in \cite{BG18b}, is a generalization of the concept of homomorphic equivalence between CSP templates. 

\begin{definition} \label{def:relax}
Let $(\relstr{A},\relstr{B})$ and $(\relstr{A}', \relstr{B}')$ be PCSP templates.
We say that $(\relstr{A}', \relstr{B}')$ is a \emph{homomorphic relaxation} of $(\relstr{A},\relstr{B})$ if there exist
homomorphisms $f: \relstr{A}' \to \relstr{A}$ and $g: \relstr{B} \to \relstr{B}'$.
\end{definition}

If $(\relstr{A}', \relstr{B}')$ is a homomorphic relaxation of $(\relstr{A},\relstr{B})$, then the trivial reduction, which does not change the input structure $\relstr{X}$, reduces  (the decision or search version of) $\PCSP(\relstr{A}',\relstr{B}')$ to  $\PCSP(\relstr{A}, \relstr{B})$. Both polynomial algorithms for $\PCSP(\OneInThree,\NAE)$ shown in the introduction come from this reduction with
$$
\relstr{A} = \relstr{B} = (A; R), \quad (x,y,z) \in R \mbox { iff } x+y+z=1\enspace,
$$
where $A = \mathbb{Z}$ in the first version of the algorithm and $A = \mathbb{Q} \setminus \{1/3\}$ in the second. In both cases, the mapping $f$ was the inclusion and the ``rounding'' mapping $g$ is defined by $g(x) = 0$ iff $x < 1/3$.

In order to define the other ingredient of a pp-construction, recall that 
a \emph{primitive positive formula} over a relational structure $\relstr{A}$ is an existentially quantified conjunction of
atomic formulas of the form $x_1 = x_2$ or $(x_{i_1}, \dots, x_{i_k}) \in R$ where $x_j$'s are variables and $R$ is a relation in $\relstr{A}$ of arity $k$.

\begin{definition}
Let $(\relstr{A},\relstr{B})$ and $(\relstr{A}' = (A'; R_1, \dots, R_n), \relstr{B}'=(B', S_1, \dots, S_n))$ be PCSP templates.

We say that $(\relstr{A}',\relstr{B}')$ is \emph{pp-definable} from $(\relstr{A},\relstr{B})$ if, for each $i \in \{1, \dots, n\}$,
there exists a primitive positive formula $\phi$ over $\relstr{A}$ such that $\phi$ defines $R_i$ and the formula, obtained by replacing each occurrence of a relation of $\relstr{A}$ by the corresponding relation in $\relstr{B}$, defines $S_i$. 

We say that $(\relstr{A}',\relstr{B}')$ is an $n$-th \emph{pp-power} of $(\relstr{A},\relstr{B})$ if $A'=A^n$, $B' = B^n$, and, if we view $k$-ary relations on $\relstr{A}'$ and $\relstr{B}'$ as $kn$-ary relations on $A$ and $B$, respectively, then $(\relstr{A}',\relstr{B}')$ is pp-definable from $(\relstr{A},\relstr{B})$.
\end{definition}

By combining these two constructions we get the notion of pp-construction.

\begin{definition}
We say that a PCSP template $(\relstr{A}, \relstr{B})$ pp-constructs a PCSP template $(\relstr{A}',\relstr{B}')$ if there exists
a sequence
$$
(\relstr{A}, \relstr{B}) = (\relstr{A}_1, \relstr{B}_1), \dots, (\relstr{A}_k, \relstr{B}_k) = (\relstr{A}', \relstr{B}')
$$
of PCSP templates such that each $(\relstr{A}_{i+1}, \relstr{B}_{i+1})$ is a pp-power or a homomorphic relaxation
of $(\relstr{A}_{i}, \relstr{B}_{i})$.
\end{definition}

It is not hard to see that if $(\relstr{A}, \relstr{B})$ pp-constructs $(\relstr{A}',\relstr{B}')$, then $\PCSP(\relstr{A}',\relstr{B}')$ reduces (even in log-space) to $\PCSP(\relstr{A},\relstr{B})$. The proof is similar to the analogous proof for CSP (see~\cite{BKW17}). An interesting alternative way for PCSP was given (but explicitly proved only for finite templates) in \cite{BKO18}.

However, in this paper, pp-constructions make only a cosmetic difference in the statement of Theorem~\ref{thm:main_fake} -- it is enough to prove the theorem for homomorphic relaxations. Indeed, it is well known (see~\cite{BOP18}) that if $(\relstr{A}, \relstr{B})$ pp-constructs $(\relstr{A}',\relstr{B}')$, then $(\relstr{A}',\relstr{B}')$ is a homomorphic relaxation of a pp-power of $(\relstr{A},\relstr{B})$.
Therefore, if a finite $\relstr{C}$ pp-constructs $(\OneInThree,\NAE)$, then $(\OneInThree,\NAE)$ is a homomorphic relaxation of a template $(\relstr{D},\relstr{D'})$, which is a pp-power of $\relstr{C}$. Then, clearly, $\relstr{D}=\relstr{D}'$ are finite and $\CSP(\relstr{D})$ reduces to $\CSP(\relstr{C})$.

\subsection{Cyclic polymorphisms}

For a PCSP template $(\relstr{A}, \relstr{B})$, a function $f: A^n \to B$ is called a \emph{polymorphism} of the template if it is a homomorphism from the $n$-th categorical power of $\relstr{A}$ to $\relstr{B}$. The basic fact of the algebraic theory of (P)CSP is that the set of polymorphisms determine the complexity of $\PCSP(\relstr{A},\relstr{B})$~(\cite{J98}, cf.~\cite{BKW17}).

We will only work with polymorphisms of CSP templates and we spell out the definition of a polymorphism in a more elementary way for this case. 

\begin{definition}
Let $\relstr{C}$ be a CSP template and $s: C^n \to C$ a function (also called an \emph{operation} in this context). 
We say that $s$ is a \emph{polymorphism} of $\relstr{C}$ if, for each relation $R$ in $\relstr{C}$ with $k=\arity(R)$ and 
all tuples $(a_1^1, \dots, a_k^1), \dots, (a_1^n, \dots, a_k^n) \in R$, we have
$$
(s(a_1^1, \dots, a_1^n), \dots, s(a_k^1, \dots, a_k^n)) \in R\enspace.
$$
\end{definition}

The proof of the main theorem is based on the following result from~\cite{BK12}. 

\begin{definition}
An operation $s: C^n \to C$ is called \emph{cyclic} if, for all $(a_1, \dots, a_n) \in C^n$, we have
$$
s(a_1,a_2, \dots, a_n) = s(a_2, \dots, a_n, a_1)\enspace.
$$
\end{definition}

\begin{theorem} \label{thm:cyclic}
Let $\relstr{C}$ be a finite CSP template. If $\CSP(\relstr{C})$ is not NP-complete, then $\relstr{C}$ has a cyclic polymorphism of arity $p$ for every prime number $p > |C|$. 
\end{theorem}

We remark that cyclic operations characterize the borderline between  NP-complete and tractable CSPs -- whenever $\relstr{C}$ has a cyclic polymorphism of arity at least 2, then $\CSP(\relstr{C})$ is tractable~\cite{Bul17,Zhuk17}. In fact, cyclic polymorphisms provide currently the strongest characterization of the borderline in the sense that the other important types of operations (such as the Sigger's operations~\cite{Sig10,KMM14} or the weak near-unanimity operations~\cite{MM08}) can be obtained from a cyclic operation by an identification of variables. The proof of Theorem~\ref{thm:main_fake} could still be simplified having a yet stronger (or alternative) characterization at hand. See Section~\ref{sec:conclusion} for a concrete open problem in this direction.

\section{Infinity is necessary}

In this section we prove the main theorem. As explained in Subsection~\ref{subsec:constr}, it is enough to prove the following result.

\begin{theorem}
Let $\relstr{C} = (C; R)$ be a finite relational structure with ternary $R \subseteq C^3$ such that $(\OneInThree, \NAE)$ is a homomorphic relaxation of $(\relstr{C},\relstr{C})$.
Then $\CSP(\relstr{C})$ is NP-complete.
\end{theorem}

Assume that $\CSP(\relstr{C})$ is not NP-complete and let $f: \OneInThree \to \relstr{C}$ and $g: \relstr{C} \to \NAE$ be homomorphisms from the definition of homomorphic relaxation, Definition~\ref{def:relax}.

Since $gf$ is a homomorphism, this mapping applied component-wise to the 1-in-3 tuple $(0,0,1)$ is a not-all-equal tuple. In particular $f(0) \neq f(1)$. We rename the elements of $C$ so that $\{0,1\} \subseteq C$ and $f$ is the inclusion. As $f$ and $g$ are homomorphisms, we get
$$
\{0,1\} \subseteq C, \quad \{(1,0,0), (0,1,0), (0,0,1)\} \subseteq R
$$
and 
$$
\neg (g(a)=g(b)=g(c)) \mbox{ whenever } (a,b,c) \in R\enspace.
$$

By Theorem~\ref{thm:cyclic}, $\relstr{C}$ has a cyclic polymorphism of any prime arity $p > |C|$. We fix a cyclic polymorphism 
$$
s \mbox{ of prime arity }
p > 60 |C|\enspace.
$$
Next we define an operation $t$ on $C$ of arity $p^2$ by
$$
t(x_{11}, x_{12}, \ldots, x_{1p}, x_{21}, x_{22}, \ldots x_{2p}, x_{31}, \ldots, \ldots, x_{pp}) \hfill
$$
\begin{align*}
 = s(&s(x_{11},x_{21}, \ldots, x_{p1}), \\
     &s(x_{12},x_{22}, \ldots, x_{p2}), \\
		 &\dots \\
		 &s(x_{1p},x_{2p}, \ldots, x_{pp}))\enspace.
\end{align*}

It will be convenient to organize the arguments of $t$ into a $p \times p$ matrix $X$ whose entry in the $i$-th row and $j$-th column is $x_{ij}$, so the value
$$
t 
\left(
\begin{array}{cccc}
x_{11} & x_{12} & \cdots & x_{1p} \\
x_{21} & x_{22} & \cdots & x_{2p} \\
\vdots & \vdots & \ddots & \vdots \\
x_{p1} & x_{p2} & \cdots & x_{pp}
\end{array}
\right)
$$
is obtained by applying $s$ to the columns and then $s$ to the results. 

We introduce several  concepts for zero-one matrices, the only important arguments of $t$ for the proof.

\begin{definition}
Let $X=(x_{ij}), Y$ be $p \times p$ zero-one matrices. The \emph{area of $X$} is the fraction of ones and is denoted
$$
\lambda(X) = \left(\sum_{i,j} x_{ij}\right)/p^2\enspace.
$$

The matrices $X,Y$ are called \emph{$g$-equivalent}, denoted $X \sim Y$, if $g(t(X)) = g(t(Y))$.

The matrix $X$ is called \emph{tame} if 
\begin{align*}
&X \sim 0_{p \times p} \quad  \mbox{ if } \quad \lambda(X) < 1/3 \\
\mbox{and } &X \sim 1_{p \times p} \quad \mbox{ if } \quad \lambda(X) > 1/3 
\end{align*}
where $0_{p \times p}$ stands for the zero matrix and $1_{p \times p}$ for the all-ones matrix.
\end{definition}

Observe that the equivalence $\sim$ has two blocks, so, e.g., $X \not\sim Y \not\sim Z$ implies $X \sim Z$. Also recall that $p>3$ is a prime number, so the area of $X$ is never equal to $1/3$.

The proof now proceeds as follows. We show that certain matrices, called ``almost rectangles'', are tame. The proof is by induction (although the proof logic, as presented, is a bit different). Subsection~\ref{subsec:lines} provides the base case and Subsection~\ref{subsec:rect} handles the induction step.
In Subsection~\ref{subsec:contra}, we construct two tame matrices $X_1$, $X_2$ such that $\lambda(X_1)<1/3$ and $\lambda(X_2) > 1/3$, but $t(X_1) = t(X_2)$ (because the corresponding columns of $X_1$ and $X_2$ will be evaluated by $s$ to the same elements). This  gives us a contradiction since $0_{p \times p} \not\sim 1_{p \times p}$ as we shall see.

\subsection{Covers}

Before launching into the proof, we introduce an additional concept and state a consequence of the fact that $s$ is a polymorphism.

\begin{definition}
A triple $X,Y,Z$ of $p \times p$ zero-one matrices is called a \emph{cover} if, for every $1 \leq i,j \leq p$, exactly one of $x_{ij},y_{ij},z_{ij}$ is equal to one.
\end{definition}

\begin{lemma} \label{lem:nae}
If $X,Y,Z$ is a cover, then $X,Y,Z$ are not all $g$-equivalent. 
\end{lemma}

\begin{proof}
By the definition of a cover, the $ij$-th coordinates of $X$, $Y$, $Z$ are in $\{(0,0,1),(0,1,0),(1,0,0)\} \subseteq R$ for each $i,j$. 
Since $t$ preserves $R$ (because $s$ does), the triple $(t(X),t(Y),t(Z))$ is in $R$ as well. 
Finally, $g$ is a homomorphism from $\relstr{C}$ to $\NAE$, therefore $g(t(X)),g(t(Y)),g(t(Z))$ are not all equal. 
In other words, $X$, $Y$, $Z$ are not all $g$-equivalent, as claimed.
\end{proof}

\subsection{Line segments are tame} \label{subsec:lines}

In this subsection it will be more convenient to regard the arguments of $t$ as a tuple $\vc{x} = (x_{11},x_{12}, \ldots)$ of length $p^2$ rather than a matrix. The concepts of the area, $g$-equivalence, tameness, and cover is extended to tuples in the obvious way.
Since $p>3$ is a prime number, $p^2$ is 1 modulo 3. Let $q$ be such that
$$
 p^2 = 3q+1\enspace.
$$
Moreover, let $\ttt{i}$ denote the following tuple of length $p^2$.
$$
\ttt{i} = (\underbrace{1,1, \cdots, 1}_{i \times }, 0,0, \cdots 0)
$$
We prove in this subsection that all such tuples are tame.
We first recall a well-known fact.

\begin{lemma}
The operation $t$ is cyclic.
\end{lemma}

\begin{proof}
By cyclically shifting the arguments we get the same result:
\begin{align*}
&t(x_{12}, \cdots, x_{pp},x_{11}) =
t 
\left(
\begin{array}{ccccc}
x_{12} & x_{13} & \cdots & x_{1p} & x_{21} \\
x_{22} & x_{23} & \cdots & x_{2p} & x_{31} \\
\vdots & \vdots & \ddots & \vdots & \vdots\\
x_{p2} & x_{p3} & \cdots & x_{pp} & x_{11} 
\end{array}
\right) \\
&= t 
\left(
\begin{array}{ccccc}
x_{21} & x_{12} & x_{13} & \cdots & x_{1p}  \\
x_{31} & x_{22} & x_{23} & \cdots & x_{2p}  \\
\vdots & \vdots & \ddots & \vdots & \vdots \\
x_{11} & x_{p2} & x_{p3} & \cdots & x_{pp}   
\end{array}
\right) \\
&=t 
\left(
\begin{array}{cccc}
x_{11} & x_{12} & \cdots & x_{1p} \\
x_{21} & x_{22} & \cdots & x_{2p}  \\
\vdots & \vdots & \ddots & \vdots \\
x_{p1} & x_{p2} & \cdots & x_{pp}  
\end{array}
\right) = 
t(x_{11}, x_{12}, \cdots, x_{pp})\enspace,
\end{align*}
where the second equality uses the cyclicity of the outer ``s'' in the definition of $t$, while the third one the cyclicity of the first inner ``s''.
\end{proof}

The following lemma is proved by induction on $i = 0,1, \dots, q$.

\begin{lemma}
For each $i \in \{0,1, \dots, q\}$, we have 
\begin{align*}
&\ttt{q-i} \sim \ttt{q-i+1} \sim \cdots \sim \ttt{q} \\ 
&\not\sim \ttt{q+1} \sim \cdots \sim \ttt{q+i} \sim \ttt{q+i+1}\enspace.
\end{align*}
\end{lemma}

\begin{proof}
For the first induction step, $i=0$,
let $\vc{x} = \ttt{q}$, let $\vc{y}$ be $\ttt{q}$ (cyclically) shifted $q$ times to the right (so the first 1 is at the $(q+1)$-st position), and let $\vc{z}$ be $\ttt{q+1}$ shifted $2q$ times to the right. 
The tuples $\vc{x},\vc{y},\vc{z}$ form a cover, therefore they are not all $g$-equivalent by Lemma~\ref{lem:nae}.
But $t$ is cyclic, thus $t(\vc{x})=t(\vc{y}) = t(\ttt{q})$ and $t(\vc{z}) = t(\ttt{q+1})$. It follows that $\ttt{q},\ttt{q},\ttt{q+1}$ are not all $g$-equivalent and we get $\ttt{q} \not\sim \ttt{q+1}$.

Now we prove the claim for $i>0$ assuming it holds for $i-1$. To verify $\ttt{q-i} \sim \ttt{q-i+1}$ consider
$\ttt{q-i}$, $\ttt{q+1}$, $\ttt{q+i}$. Since $(q-i)+(q+1)+(q+i) = 3q+1=p^2$, these tuples can be shifted to form a cover and then the same argument as above gives us that $\ttt{q-i}$, $\ttt{q+1}$, $\ttt{q+i}$ are not all $g$-equivalent. But $\ttt{q+1} \sim \ttt{q+i}$ by the induction hypothesis, therefore $\ttt{q-i} \not\sim \ttt{q+1}$. Since $\ttt{q+1} \not\sim \ttt{q-i+1}$ (again by the induction hypothesis), we get $\ttt{q-i} \sim \ttt{q-i+1}$, as required.

It remains to check $\ttt{q+i} \sim \ttt{q+i+1}$. This is done in a similar way, using the tuples $\ttt{q-i}$, $\ttt{q}$, $\ttt{q+i+1}$. 
\end{proof}

We have proved that $\ttt{0} \sim \dots \sim \ttt{q} \not\sim \ttt{q+1} \sim \dots \sim \ttt{2q+1}$.
Using the same argument as in the previous lemma once more for $\ttt{0},\ttt{p^2-i},\ttt{i}$ with $p^2 \geq i > 2q+1$ we get $\ttt{i} \not\sim \ttt{0}$. In summary,
$\ttt{i} \sim \ttt{0}$ whenever $i \leq q$ and $\ttt{i} \sim \ttt{p^2} \not\sim \ttt{0}$ when $i \geq q+1$. 
Observing that $\lambda(\ttt{i})<1/3$ iff $i \leq q$ we obtain the following lemma.

\begin{lemma} \label{lem:lines_tame}
Each $\ttt{i}$, $i \in \{0, 1, \cdots, p^2\}$, is tame and $\ttt{0} \not\sim \ttt{p^2}$. 
\end{lemma}

\subsection{Almost rectangles are tame} \label{subsec:rect}

We start by introducing a special type of zero-one matrices. 

\begin{definition}
Let $1 \leq k_1, \dots, k_p \leq p$. By
$$
\rrr{k_1,k_2,\dots,k_p}
$$
we denote the matrix whose $i$-th column begins with $k_i$ ones followed by ($p-k_i$) zeros, for each $i \in \{1, \dots, p\}$.

An \emph{almost rectangle} is a matrix of the form $\rrr{k,k, \dots, k, l, l, \dots, l}$ (the number of $k$'s can be arbitrary, including 0 or $p$) where $0 \leq k-l \leq 5|C|$. The quantity $k-l$ is referred to as the \emph{size of the step}. 
\end{definition}

In the remainder of this subsection we prove the following proposition.

\begin{proposition} \label{prop:tame}
Each almost rectangle is tame.
\end{proposition}

Let 
$$
X=\rrr{\underbrace{k, k, \cdots, k}_{m \times}, l,l, \dots, l}
$$
be a minimal counterexample in the following sense.
\begin{itemize}
\item $X$ has the minimum size of the step and,
\item among such counterexamples, $\abs{\lambda(X)-1/3}$ is maximal.
\end{itemize}

\begin{lemma}
The size of the step of $X$ is at least 2.
\end{lemma}

\begin{proof}
This lemma is just a different formulation of Lemma~\ref{lem:lines_tame} since an almost rectangle with step of size 0 or 1 represents the same choice of arguments as $\ttt{i}$ for some $i$. 
\end{proof}

We handle two cases $\lambda(X) \geq 5/12$ and $\lambda(X) \leq 5/12$ separately, but the basic idea for both of them is the same as in the proof of Lemma~\ref{lem:lines_tame}. To avoid puzzles, let us remark that any number strictly between $1/3$ and $1/2$ (instead of $5/12$) would work with a sufficiently large $p$. 

\begin{lemma} \label{lem:c}
The area of $X$ is less than $5/12$.
\end{lemma}

\begin{proof}
Assume that $\lambda(X) \geq 5/12$. Let $k_1$, $k_2$, $l_1$, and $l_2$ be the non-negative integers such that 
\begin{eqnarray}
l_1+l_2+k = p = k_1+k_2+l, \label{eq:a} \\
1 \geq k_1 - k_2 \geq 0, \mbox{ and }
1 \geq l_1-l_2 \geq 0 \enspace. \label{eq:b}
\end{eqnarray} 
We have $k_1 \geq l_1$ and $k_2 \geq l_2$. Moreover,
since $k-l \geq 2$ by the previous lemma, it follows that both $k_1-l_1$ and $k_2-l_2$ are strictly smaller than $k-l$.

Consider the matrices
$$
Y_i = \rrr{\underbrace{l_i,l_i, \dots, l_i}_{m \times}, k_i, k_i, \dots, k_i}, \quad i=1,2\enspace.
$$
By shifting all the rows of $Y_i$, $i \in \{1,2\}$, $m$ times to the left we obtain an almost rectangle with a smaller step size than $X$, which is thus tame by the minimality assumption on $X$. 
Since such a shift changes neither the value of $t$ (as the outer ``$s$'' in the definition of $t$ is cyclic) nor the area, both $Y_1$ and $Y_2$ are tame matrices.

Let $Y_1'$ ($Y_2'$, resp.) be the matrices obtained from $Y_1$ ($Y_2$, resp.) by shifting the first $m$ columns $k$ times ($k+l_1$ times, resp.) down and the remaining columns $l$ times ($l+k_1$ times, resp.) down. Since $X,Y_1',Y_2'$ is a cover (by~(\ref{eq:a})) and
cyclically shifting columns does not change the value of $t$ (as the inner occurrences of ``$s$'' in the definition of $t$ are cyclic), Lemma~\ref{lem:nae} implies that $X$, $Y_1$, $Y_2$ are not all $g$-equivalent. 

From $X, Y_1', Y_2'$ being a cover, it also follows that
$$
\lambda(X) + \lambda(Y_1') + \lambda(Y_2') = \lambda(X) + \lambda(Y_1) + \lambda(Y_2) = 1\enspace.
$$
Moreover, by~(\ref{eq:b}), we have $\lambda(Y_2) \leq \lambda(Y_1)$ and these areas differ by at most $p/p^2=1/p$. Therefore
$$
\lambda(Y_1) = 1 - \lambda(X) - \lambda(Y_2) \leq 1 - 5/12 - \lambda(Y_1) + 1/p
$$
and, since $p > 12$ by the choice of $p$, we obtain
$$
\lambda(Y_2) \leq \lambda(Y_1) < 1/3\enspace.
$$

The tameness of $Y_i$ now gives us $Y_1 \sim Y_2 \sim 0_{p \times p}$ and then, since $Y_1,Y_2,X$ are not all $g$-equivalent and $0_{p \times p} \not\sim 1_{p \times p}$ (by the second part of Lemma~\ref{lem:lines_tame}), 
we get $X \sim 1_{p \times p}$. But $\lambda(X) \geq 5/12 > 1/3$, hence $X$ is tame, a contradiction with the choice of $X$. 
\end{proof}

It remains to handle the case $\lambda(X) < 5/12$. 

We first claim that $2k$ (and thus $k+l$ and $2l$) is less than $p$. Indeed, since the step size of $X$ is at most $5|C|$ (by the definition of an almost rectangle) and $p > 60|C|$, we get
\begin{align*}
5/12 > \lambda(X) &\geq \frac{p(k - 5|C|)}{p^2} \\
k & \leq 5p/12 + 5|C| < 5p/12 + p/12 = p/2\enspace.
\end{align*}

We now again need to distinguish two cases.
Assume first that $m < p/2$. 

Let 
\begin{align*}
Y &= \rrr{\underbrace{l, \cdots,l}_{m \times },\underbrace{k, \cdots, k}_{m \times }, l, \cdots, l},  \\
Z &= \rrr{\underbrace{p-k-l, \cdots, p-k-l}_{2m \times}, p-2l, \cdots, p-2l} \enspace.
\end{align*}
The definition of $Z$ makes sense since $p-k-l, p-2l \geq 0$ by the inequality $2k < p$ derived above. 

The triple $X,Y,Z$ (similarly to $X,Y_1,Y_2$ in the proof of Lemma~\ref{lem:c}) is such that we can obtain a cover by shifting the columns down. Therefore $X$, $Y$, $Z$ are not all $g$-equivalent and $\lambda(X)+\lambda(Y) + \lambda(Z)=1$.

On the other hand, by shifting all the rows of $Y$ $m$ times to the left we obtain $X$.
We get $\lambda(X) = \lambda(Y)$ and $t(X)=t(Y)$, therefore  $Z \not\sim X$ by the previous paragraph. 

Moreover,
 by shifting all the rows of $Z$ $2m$ times to the left we obtain an almost rectangle $Z'$ with $t(Z)=t(Z')$ and $\lambda(Z)=\lambda(Z')$. The step size of $Z'$ is $(p-2l) - (p-k-l) = k-l$, which is the same as the step size of $X$. However, the distance of its area from $1/3$ is strictly greater as shown by the following calculation. 
\begin{align*}
\frac{\abs{\lambda(Z)-1/3}}{\abs{\lambda(X)-1/3}} &=
\frac{\abs{(1 - 2\lambda(X)) - 1/3}}{\abs{\lambda(X)-1/3}} \\ &=
\frac{\abs{2(1/3 - \lambda(X))}}{\abs{\lambda(X)-1/3}} = 2 > 1\enspace.
\end{align*}
By the minimality of $X$, the almost rectangle $Z'$ is tame and so is $Z$. It is also apparent from the calculation that the signs of $\lambda(X)-1/3$ and $\lambda(Z)-1/3$ are opposite. Combining these two facts with $Z \not\sim X$ derived above, we obtain that $X$ is tame, a contradiction.

In the other case, when $m > p/2$, the proof is similar using the tuples
\begin{align*}
Y &= (l, \cdots, l, \underbrace{k, \cdots, k}_{m \times}), \\
Z &= (\underbrace{p-k-l, \cdots, p-k-l}_{(p-m) \times}, p-2k, \cdots, p-2k, \\
  &\quad \underbrace{p-k-l, \cdots, p-k-l}_{(p-m) \times}) \enspace.
\end{align*}
The proof of Proposition~\ref{prop:tame} is concluded.

\subsection{Contradiction} \label{subsec:contra}

Let 
$$
m = (p-1)/2
$$
and choose natural numbers $l_1$ and $l_2$ so that
$$
p/3 - 2|C| < l_1 < l_2 < p/3 
$$
and
$$
s(\underbrace{1, \cdots, 1}_{l_1 \times}, 0, \cdots, 0)
=
s(\underbrace{1, \cdots, 1}_{l_2 \times}, 0, \cdots, 0)\enspace.
$$
This is possible by the pigeonhole principle since there are $2|C| > C$ integers in the interval and $p/3 - 2|C| > 0$ by the choice of  $p$.

The sought after contradiction will be obtained by considering the two matrices
$$
X_i = \rrr{ \underbrace{k, \dots, k}_{m \times}, l_i, \dots, l_i}, \ i=1,2 \enspace,
$$ 
where $k$ will be specified soon.

Before choosing $k$, we observe that $t(X_1) = t(X_2)$. Indeed, the first $m$ columns of these matrices are the same (and thus so are their images under $s$) and the remaining columns have the same image under $s$ by the choice of $l_1$ and $l_2$. The claim thus follows from the definition of $t$. 

Next, note that for $k \leq p/3$ the area of both matrices is less than $1/3$ since $l_i < p/3$. On the other hand, for $k \geq p/3 + 3|C|$ the area is greater:
\begin{align*}
\lambda(X_i) 
&= \frac{mk+(p-m)l_i}{p^2} \\
&\geq \frac{\frac{p-1}{2} (p/3 + 3|C|) + \frac{p+1}{2} (p/3-2|C|)}{p^2} \\
&= \frac{p^2/3 + |C| (p-5)/2}{p^2} > 1/3\enspace.
\end{align*}   

Choose the maximum $k$ so that $\lambda(X_1) < 1/3$. The derived inequalities and the choice of $l_i$ implies
$$
l_1 < l_2 \leq k < p/3 + 3|C| \leq l_1 + 5|C| < l_2 + 5|C| \enspace,
$$
therefore both $X_1$ and $X_2$ are almost rectangles. By Proposition~\ref{prop:tame}, $X_1$ and $X_2$ are tame.

Since the area of $X_1$ is less than $1/3$, we get $X_1 \sim 0_{p \times p}$.   
We chose $k$ so that increasing $k$ by 1 makes the area of $X_1$ greater than $1/3$. 
From $m < p/2$ it follows that increasing $l_1$ by 1 makes the area even greater, hence $\lambda(X_2) > 1/3$ (recall that $l_2>l_1$)
and we obtain $X_2 \sim 1_{p \times p}$. 

Recall that $0_{p \times p} \not\sim 1_{p \times p}$ by the second part of Lemma~\ref{lem:lines_tame}.
Therefore $X_1 \not\sim X_2$, contradicting $t(X_1) = t(X_2)$.

\section{Conclusion} \label{sec:conclusion}

This paper shows that if $\OneInThree \to \relstr{C} \to \NAE$ and $\relstr{C}$ is finite, then $\CSP(\relstr{C})$ is NP-complete. The proof strategy is based on Theorem~\ref{thm:cyclic} and a simple fact that, given $\relstr{A} \to \relstr{C} \to \relstr{B}$,  each polymorphism of $\relstr{C}$ induces a polymorphism of $(\relstr{A},\relstr{B})$ (by composition with the homomorphism $\relstr{A} \to \relstr{C}$ from the inside and with $\relstr{C} \to \relstr{B}$ from the outside). 

There is an algebraic sufficient condition for NP-hardness for all $\omega$-categorical structures~\cite{BarP16} -- $\CSP(\relstr{C})$ is NP-hard whenever $\relstr{C}$ does not have a \emph{pseudo-Siggers} polymorphism, that is, a 6-ary polymorphism $s$ such that
$$
\alpha s(x,y,x,z,y,z) = \beta s(y,x,z,x,z,y) \mbox{ for all } x,y,z \in C \enspace,
$$
where $\alpha$ and $\beta$ are unary polymorphisms of $\relstr{C}$. Is it possible to apply pseudo-Siggers operations to strengthen the main theorem?

\begin{question}
Let $\relstr{C}$ be an $\omega$-categorical structure that pp-constructs $(\OneInThree, \NAE)$.
Is $\CSP(\relstr{C})$ necessarily NP-hard?
\end{question}  

The proof of Theorem~\ref{thm:main_fake} could be simplified if we had stronger or more suitable polymorphisms than cyclic operations. Alternative versions of Theorem~\ref{thm:cyclic} could also help in simplifying the proof of the CSP dichotomy conjecture. 
In particular, the following question seems open.
\begin{question}
Let $\relstr{C}$ be a finite relational structure with a cyclic polymorphism of arity at least 2. 
Does $\relstr{C}$ necessarily have a polymorphism $s$ of arity $n > 1$ such that, for any $a,b \in C$ and $(x_1, \dots, x_n) \in \{a,b\}^n$, the value $s(x_1, \dots, x_n)$ depends only on the number of occurrences of $a$ in $(x_1, \dots, x_n)$?
\end{question}
Note that a more optimistic version involving evaluations with $|\{x_1, \dots, x_n\}|=3$ is disproved by considering the polymorphisms of the disjoint union of a directed 2-cycle and a directed 3-cycle.

Let us finish with an optimistic outlook. While the main result of this paper is negative, its message is rather positive. It suggests that algebraic and analytical methods in the finite-domain CSP and PCSP should be combined with the model theoretic methods used for the infinite domains, and such a combination promises a significant synergy gain.

\bibliographystyle{plain}
\bibliography{CSPbib}

\def\cprime{$'$}
\begin{thebibliography}{10}

\bibitem{AGH17}
Per Austrin, Venkatesan Guruswami, and Johan H{\aa }stad.
\newblock {$(2+\varepsilon)$}-{S}at is {NP}-hard.
\newblock {\em SIAM J. Comput.}, 46(5):1554--1573, 2017.

\bibitem{BK12}
Libor Barto and Marcin Kozik.
\newblock Absorbing subalgebras, cyclic terms, and the constraint satisfaction
  problem.
\newblock {\em Logical Methods in Computer Science}, 8(1), 2012.

\bibitem{BKW17}
Libor Barto, Andrei Krokhin, and Ross Willard.
\newblock {Polymorphisms, and How to Use Them}.
\newblock In Andrei Krokhin and Stanislav Zivny, editors, {\em The Constraint
  Satisfaction Problem: Complexity and Approximability}, volume~7 of {\em
  Dagstuhl Follow-Ups}, pages 1--44. Schloss Dagstuhl--Leibniz-Zentrum fuer
  Informatik, Dagstuhl, Germany, 2017.

\bibitem{BOP18}
Libor Barto, Jakub Opr{\v{s}}al, and Michael Pinsker.
\newblock The wonderland of reflections.
\newblock {\em Israel Journal of Mathematics}, 223(1):363--398, 2018.

\bibitem{BarP16}
Libor Barto and Michael Pinsker.
\newblock The algebraic dichotomy conjecture for infinite domain constraint
  satisfaction problems.
\newblock In {\em Proceedings of the 31st Annual ACM/IEEE Symposium on Logic in
  Computer Science}, LICS '16, pages 615--622, New York, NY, USA, 2016. ACM.

\bibitem{Bod08}
Manuel Bodirsky.
\newblock Constraint satisfaction problems with infinite templates.
\newblock In Nadia Creignou, Phokion~G. Kolaitis, and Heribert Vollmer,
  editors, {\em Complexity of Constraints}, volume 5250 of {\em Lecture Notes
  in Computer Science}, pages 196--228. Springer, 2008.

\bibitem{BodG08}
Manuel Bodirsky and Martin Grohe.
\newblock Non-dichotomies in constraint satisfaction complexity.
\newblock In Luca Aceto, Ivan Damgard, Leslie~Ann Goldberg, Magn\'us~M.
  Halld\'orsson, Anna Ing\'olfsd\'ottir, and Igor Walukiewicz, editors, {\em
  Automata, Languages and Programming}, Lecture Notes in Computer Science,
  pages 184 --196. Springer Verlag, 2008.

\bibitem{BodM17}
Manuel Bodirsky and Marcello Mamino.
\newblock {Constraint Satisfaction Problems over Numeric Domains}.
\newblock In Andrei Krokhin and Stanislav Zivny, editors, {\em The Constraint
  Satisfaction Problem: Complexity and Approximability}, volume~7 of {\em
  Dagstuhl Follow-Ups}, pages 79--111. Schloss Dagstuhl--Leibniz-Zentrum fuer
  Informatik, Dagstuhl, Germany, 2017.

\bibitem{BG16}
Joshua Brakensiek and Venkatesan Guruswami.
\newblock New hardness results for graph and hypergraph colorings.
\newblock In {\em Proceedings of the 31st Conference on Computational
  Complexity}, CCC '16, pages 14:1--14:27, Germany, 2016. Schloss
  Dagstuhl--Leibniz-Zentrum fuer Informatik.

\bibitem{BG18b}
Joshua Brakensiek and Venkatesan Guruswami.
\newblock An algorithmic blend of {LP}s and ring equations for promise {CSP}s.
\newblock {\em CoRR}, abs/1807.05194, 2018.
\newblock to appear in SODA'19.

\bibitem{BG18a}
Joshua Brakensiek and Venkatesan Guruswami.
\newblock Promise constraint satisfaction: Structure theory and a symmetric
  boolean dichotomy.
\newblock In {\em Proceedings of the Twenty-Ninth Annual ACM-SIAM Symposium on
  Discrete Algorithms}, SODA '18, pages 1782--1801, Philadelphia, PA, USA,
  2018. Society for Industrial and Applied Mathematics.

\bibitem{Bul17}
A.~A. Bulatov.
\newblock A dichotomy theorem for nonuniform {CSP}s.
\newblock In {\em 2017 IEEE 58th Annual Symposium on Foundations of Computer
  Science (FOCS)}, volume~00, pages 319--330, 2017.

\bibitem{BJK05}
Andrei Bulatov, Peter Jeavons, and Andrei Krokhin.
\newblock Classifying the complexity of constraints using finite algebras.
\newblock {\em SIAM J. Comput.}, 34:720--742, 2005.

\bibitem{BKO18}
Jakub Bul{\'{\i}}n, Andrei~A. Krokhin, and Jakub Opr{\v s}al.
\newblock Algebraic approach to promise constraint satisfaction.
\newblock {\em CoRR}, abs/1811.00970, 2018.

\bibitem{FV98}
Tom\'{a}s Feder and Moshe~Y. Vardi.
\newblock The computational structure of monotone monadic {SNP} and constraint
  satisfaction: A study through datalog and group theory.
\newblock {\em SIAM Journal on Computing}, 28(1):57--104, 1998.

\bibitem{GLS93}
Martin Gr\"{o}tschel, L\'{a}szl\'{o} Lov\'{a}sz, and Alexander Schrijver.
\newblock {\em Geometric algorithms and combinatorial optimization}, volume~2
  of {\em Algorithms and Combinatorics}.
\newblock Springer-Verlag, Berlin, second edition, 1993.

\bibitem{H01}
Johan H{\aa}stad.
\newblock Some optimal inapproximability results.
\newblock {\em J. ACM}, 48:798--859, 2001.

\bibitem{HN90}
Pavol Hell and Jaroslav Ne{\v{s}}et{\v{r}}il.
\newblock On the complexity of {$H$}-coloring.
\newblock {\em J. Combin. Theory Ser. B}, 48(1):92--110, 1990.

\bibitem{J98}
Peter Jeavons.
\newblock On the algebraic structure of combinatorial problems.
\newblock {\em Theoretical Computer Science}, 200(1--2):185 -- 204, 1998.

\bibitem{JCG97}
Peter Jeavons, David Cohen, and Marc Gyssens.
\newblock Closure properties of constraints.
\newblock {\em J. ACM}, 44(4):527--548, 1997.

\bibitem{KMM14}
Keith Kearnes, Petar Markovi\'c, and Ralph McKenzie.
\newblock Optimal strong {M}al'cev conditions for omitting type 1 in locally
  finite varieties.
\newblock {\em Algebra universalis}, 72(1):91--100, 2014.

\bibitem{MM08}
Mikl{\'o}s Mar{\'o}ti and Ralph McKenzie.
\newblock Existence theorems for weakly symmetric operations.
\newblock {\em Algebra Universalis}, 59(3-4):463--489, 2008.

\bibitem{Pin15}
Michael {Pinsker}.
\newblock {Algebraic and model theoretic methods in constraint satisfaction}.
\newblock {\em arXiv e-prints}, page arXiv:1507.00931, 2015.

\bibitem{Sch78}
Thomas~J. Schaefer.
\newblock The complexity of satisfiability problems.
\newblock In {\em Conference {R}ecord of the {T}enth {A}nnual {ACM} {S}ymposium
  on {T}heory of {C}omputing ({S}an {D}iego, {C}alif., 1978)}, pages 216--226.
  ACM, New York, 1978.

\bibitem{Sig10}
Mark~H. Siggers.
\newblock A strong {M}al'cev condition for locally finite varieties omitting
  the unary type.
\newblock {\em Algebra universalis}, 64(1-2):15--20, 2010.

\bibitem{Zhuk17}
D.~Zhuk.
\newblock A proof of {CSP} dichotomy conjecture.
\newblock In {\em 2017 IEEE 58th Annual Symposium on Foundations of Computer
  Science (FOCS)}, volume~00, pages 331--342, 2017.

\end{thebibliography}

\end{document}